\documentclass[11pt,titlepage]{article}
\usepackage{amsthm}

\newtheorem{proposition}{Proposition}

\newtheorem{theorem}{Theorem}

\usepackage[latin1]{inputenc}
\usepackage{rotating, rotfloat}
\usepackage[round]{natbib} 
\usepackage{graphics, graphicx, hyperref, url, amsmath, subfigure, amsthm, amsfonts, enumerate, epstopdf, booktabs, bm}
\hypersetup{
    pdfstartview={FitH},    
    colorlinks=true,        
    linkcolor=blue,         
    citecolor=blue,         
    filecolor=blue,         
    urlcolor=blue           
}
 \def\x{\textbf{\textit{x}}}  \def\vv{\textbf{\textit{v}}}
\def\y{\textbf{\textit{y}}}  \def\b{\textbf{\textit{b}}} 
  \def\0{\boldsymbol 0}       
  
 \def\p{\textbf{\textit{p}}} 

\def\etal{\itshape et al.}

\newcommand{\LD}{{\mbox{\bf LD}}}

\newcommand{\MD}{{\mbox{\bf MD}}}
\newcommand{\ZD}{{\mbox{\bf ZD}}}

\newcommand{\EZD}{{\mbox{\bf EZD}}}

\newcommand{\bec}{\begin{center}}
\newcommand{\enc}{\end{center}}
\newcommand{\bee}{\begin{eqnarray*}}
\newcommand{\ene}{\end{eqnarray*}}
\newcommand{\beq}{\begin{equation}}
\newcommand{\eeq}{\end{equation}}

\begin{document}

\title{\bf A new class of $L_q$-norm zonoid depths}
\vskip 5mm

\author {{Xiaohui Liu$^{a, b}$   \footnote{Corresponding author's email: csuliuxh912@gmail.com.}
        Yuanyuan Li$^{a, b}$, Qing Liu$^{a, b}$
        }\\ \\[1ex]
        {\em\footnotesize $^a$ School of Statistics, Jiangxi University of Finance and Economics, Nanchang, Jiangxi 330013, China}\\
        {\em\footnotesize $^b$ Research Center of Applied Statistics, Jiangxi University of Finance and Economics, Nanchang,}\\ {\em\footnotesize Jiangxi 330013, China}\\
}

\date{}
\maketitle

\begin{center}
{\sc Abstract}
\end{center}

Zonoid depth, as a well-known ordering tool, has been widely used in multivariate analysis. However, since its depth value vanishes outside the convex hull of the data cloud, it suffers from the so-called `outside problem', which consequently hinders its many practical applications, e.g., classification. In this note, we propose a new class of \emph{$L_q$-norm zonoid depths}, which has no such problem. Examples are also provided to illustrate their contours.

\vspace{2mm}

{\small {\bf\itshape Key words:} $L_q$-norm zonoid depth; Outside problem; Zonoid depth; Mahalanobis depth}
\vspace{2mm}

{\small {\bf2000 Mathematics Subject Classification Codes:} 62F10; 62F40; 62F35}

\setlength{\baselineskip}{1.5\baselineskip}

\vskip 0.1 in
\section{Introduction}
\paragraph{}
\vskip 0.1 in \label{Introduction}

To provide multivariate observations a proper ordering, \cite{Tuk1975} suggested to so-called halfspace depth function. This depth enjoys many desirable properties, and can be utilized to extend the univariate median, as well as many other location estimators, to higher dimensions. Many other depth functions are also proposed motivated by \cite{Tuk1975} from different principle. The most famous notions include the simplicial depth \citep{Liu1990}, and projection depth \citep{Zuo2003}.

These three depth notions satisfy all four properties, namely, (a) \emph{affine-invariance}, (b) \emph{maximality at a center point}, (c) \emph{monotonicity related to the center point}, and (d) \emph{vanishing at infinity}, of defining a general statistical depth function \citep{ZS2000}. They all maximize at the medians, which all deduce to the ordinary median in the univariate case. In this sense, the statistical inferential procedures induced from these depths enjoy somewhat some robust properties.

Furthermore, in statistics, there exist some other depths, which are \emph{not robust} in the sense that they maximize at the conventional mean. But since they can fit to some other different applications, they are still of practical importance.

Among them, the simplest one is the Mahalanobis depth \citep{ZS2000}, which is constructed on the Mahalanobis distance of $\x$ to the sample mean $\bar{X}$. Giving the random sample $\mathcal{X}^{n} = \{X_1, X_2, \cdots, X_n\} \subset \mathbb{R}^d$ ($d \ge 1$), its definition is as follows:
\begin{eqnarray*}
  \MD(\x, P_n) = \frac{1}{1 + \sqrt{(\x - \bar{X})^\top \hat{\Sigma}_n^{-1} (\x - \bar{X})}},
\end{eqnarray*}
where $\hat{\Sigma}_n = \frac{1}{n} \sum\limits_{i=1}^n (X_i - \bar{X}) (X_i - \bar{X})^\top$, and $P_n$ denotes the empirical probability measure related to $\mathcal{X}^n$. However, the contours induced from this depth is always of elliptical shape, even though the data are generated from a \emph{non-elliptical} distribution. It hence fails to characterize some special properties of the data cloud, although it is both computationally and conceptually simple.

An another famous depth of this type is the zonoid depth suggested by \cite{KM1997}. Its empirical version is as follows:
\begin{eqnarray}
\label{Def:ZD}
  \ZD(\x, P_n) &=&
  \begin{cases}
    \sup\left\{ \alpha: \x = \sum\limits_{i=1}^n p_i X_i, ~\sum\limits_{i=1}^n p_i = 1, ~np_i \in [0, 1/\alpha], ~\forall i \right\}, & \x \in \textbf{conv}(\mathcal{X}^n)\\[3ex]
    0, & \x \notin \textbf{conv}(\mathcal{X}^n),
  \end{cases}
\end{eqnarray}
where $\textbf{conv}(\mathcal{X}^n)$ denotes the convex hull of the data cloud $\mathcal{X}^n$. Different from the Mahalanobis depth above, zonoid depth induces data-dependent contours, and can characterize the underlying distribution in a unique way.

Unfortunately, this depth suffers from the so-called `outside problem', i.e., its depth vanishes outside the convex hull of the data cloud. As a result, it is \emph{impossible} to classify a point outside all convex hulls of the classes into a specific class when using the zonoid depth in classification \citep{LMM2014}. Recently, \cite{Liu2017} suggested a new class of general empirical depth, including the zonoid depth as a special case. Nevertheless, they still suffer from the outside problem.

Efforts to overcome this problem include using a combination of zonoid and Mahalanobis depth, i.e., zonoid-Mahalanobis depth, suggested by \cite{Hob2003}. Its definition is as follows:
\begin{eqnarray*}
  D(\x, P_n) = \max\left\{\ZD(\x, P_n), \beta \MD(\x, P_n)\right\} \text{ with } \beta = \frac{1}{\max\limits_{1\leq i \leq n} p_i},
\end{eqnarray*}
for some $p_i$'s specified in \ref{Def:ZD}. However, this solution unnaturally requires combining two depth functions, and more importantly, the contours induced from this depth is not convex as shown in \cite{MH2006}, which may bring difficulty to its computation.

How to further improve this is not trivial. 
In the sequel we propose a new class of $L_q$-norm zonoid depths. It turns that the new depths satisfies all four properties of defining a general statistical function as suggested in \cite{ZS2000}. Their depth values do not vanish even outside the convex hull of the data cloud, behaving similarly to the projection depth \citep{Zuo2003}. Nevertheless, since $L_q$-norm zonoid depths do not involve the methodology of projection pursuit, their computation is easy to achieve, while the computation of the projection depth is quite intensive \citep{LZ2014}. When $q = 2$, the $L_q$-norm zonoid depths deduce to the Mahalanobis depth, and hence they include the so-called Mahalanobis depth as a special case. Furthermore, by taking limit to $q$, we obtain the degenerate case $L_\infty$--norm zonoid depth corresponding to $q = +\infty$, whose form is very close to the conventional zonoid depth developed by \cite{KM1997}. The contours induced from these $L_q$-norm zonoid depths for any $q \ge 1$ are still convex and nested, and center at the conventional sample mean. Hence, these depth notions may serve as an alterative to the zonoid depth if the outside problem is a big concern having to be taken into account.

The rest of this note is organized as follows. We present the methodology and main results of this note in Section \ref{Sec:MMS}. Some illustrative examples are given in Section \ref{Sec:Illustrations}. Concluding remarks end this note.

\vskip 0.1 in
\section{Methodology and main results}
\paragraph{}
\vskip 0.1 in \label{Sec:MMS}

In this section, we first explore the idea behind an equivalent definition to the zonoid depth given \cite{KM1997}. Then we suggest a new class of $L_q$-norm zonoid depths, which do not vanish outside the convex hull of the data cloud. The population versions are also derived under some mild conditions.

Let's start with definition of the zonoid depth. For the original zonoid depth given in \eqref{Def:ZD}, \cite{KM1997} derived also the following equivalent definition:
\begin{eqnarray*}
  \ZD(\x, P_n) &=& \sup\left\{ \left(\max\limits_{1\le i \le n} (np_i)\right)^{-1}: \x = \sum_{i=1}^n p_i X_i, ~(p_1, p_2, \cdots, p_n) \in \mathbb{S}^{n-1}\right\}\\
  &=& \left(\inf\left\{\max\limits_{1\le i \le n} (np_i): (p_1, p_2, \cdots, p_n) \in \mathbb{S}_x \right\} \right)^{-1},
\end{eqnarray*}
when $\x \in \textbf{conv}(\mathcal{X}^n)$, where $\mathbb{S}_x = \{\p \in \mathbb{S}^{n-1}: \x = \sum_{i=1}^n p_i X_i = \mathbf{A}_X \p\}$ with $\mathbf{A}_X = (X_1, X_2, \cdots, X_n)$, and $\mathbb{S}^{n-1} = \{\p = (p_1, p_2, \cdots, p_n)^\top: ~\sum_{i=1}^n p_i = 1, ~ p_i \ge 0, \forall i\}$.

This equivalent definition indicates that the zonoid depth value of $\x$ depends on the similarity between $\p = (p_1, p_2, \cdots, p_n)^\top \in \mathbb{S}_x$ and the center $\p_0 = (\frac{1}{n}, \frac{1}{n}, \cdots, \frac{1}{n})^\top$ of $\mathbb{S}^{n-1}$, and we are using the most similar point $\p$ to $\p_0$ on $\mathbb{S}_x$ to characterize the deepness of $\x$ with respect to $\mathcal{X}^n$ \citep{Liu2017}. Obviously, the values of $p_i$'s involved in $\max\limits_{1\le i \le n} (np_i)$ can not be negative, because it may consequently make the depth value of $\x$ to be negative. As a result, $\x$ should be contained in the convex hull of $\mathcal{X}^n$ if one wants $\ZD(\x, P_n) > 0$. In turn, the outside problem exists. Although using instead $\max\limits_{1\le i \le n} |np_i|$ in the definition can avoid this `negative value' problem, it is easy to check that the depth function based on $\max\limits_{1\le i \le n} |np_i|$ does not satisfy all four properties of defining a general depth function any more, nevertheless. Hence, additional efforts are still in need to overcome the so-called outside problem.

Observe that one may also use the \emph{distance} between $\p$ and $\p_0$ to measure their similarity. In fact, for any $\p \in \mathbb{S}^{n - 1}$, since
\begin{eqnarray*}
  \max_{1\leq i \leq n} (n p_i) &=& \max_{1 \leq i \leq n} (np_i - 1 + 1)\\
  &=& 1 + \max_{1 \leq i \leq n} (np_i - 1),
\end{eqnarray*}
the zonoid depth can be further reexpressed as
\begin{eqnarray*}
  \ZD(\x, P_n) &=& \sup\left\{\frac{1}{1 + \max\limits_{1 \leq i \leq n} (np_i - 1)}: \p \in \mathbb{S}_x \right\}\\
  &=& \frac{1}{1 + \inf\limits_{\p \in \mathbb{S}_x} \left\{\max\limits_{1 \leq i \leq n} (np_i - 1) \right\}}.
\end{eqnarray*}
Hence, if adding an absolute sign to $np_i - 1$, we may enlarge the domain $\mathbb{S}_x$ of $\p$ to the hyperplane $\mathbb{P}_x$, and possibly still can obtain a well defined depth function, because $\mathbb{P}_x$ includes the convex polytope $\mathbb{S}_x$ as a subset, and we are using the most similar point to $\p_0$ in $\mathbb{P}_x$ to define the depths. Here $\mathbb{P}_x = \{\p \in \mathbb{R}^{n}: \mathbf{A}_x \p = \x, ~ \textbf{1}_n^\top \p = 1\}$ with $\textbf{1}_n = (1, 1, \cdots, 1)^\top$ being the vector of $n$ ones.

Motivated by this, we propose to consider the following $L_\infty$ zonoid depth:
\begin{eqnarray}
\label{eqn:EZD}
  \LD_\infty(\x, P_n) = \frac{1}{1 + S(\x, P_n)}.
\end{eqnarray}
where
\begin{eqnarray*}
  S(\x, P_n) = \inf\limits_{\p \in \mathbb{P}_x} d_{\infty} (n\p, n\p_0).
\end{eqnarray*}
Since $\mathbb{P}_x$ does \emph{not} necessarily require the components of $\p$ to be positive, the value of $\LD_\infty(\x, P_n)$ can be well defined for some $\x$ even outside of the convex hull of the data cloud.

Here we still use the term `zonoid' is because the definition of the $L_\infty$ zonoid depth involves the \emph{lift zonoid} as did in \cite{KM1997}:
\begin{eqnarray*}
  \{z(P, g): g:\mathbb{R}^d \rightarrow R^1 \text{ measurable}\}
\end{eqnarray*}
for given probability measure $P$, where $z(P, g) = (z_0(P, g), \zeta(P, g)^\top)^\top \in \mathbb{R}^{d+1}$ with
\begin{eqnarray*}
  \begin{cases}
    z_0(P, g) = \int g(X) dP\\[3ex]
    \zeta(P, g) = \int Xg(X) dP.
  \end{cases}
\end{eqnarray*}

For the depth given in \eqref{eqn:EZD}, the following theorem indicates that it is a well defined statistical depth function. Hence, it can also be used to provide a desirable center-outward ordering for multivariate observations like what the zonoid depth of \cite{KM1997}, as well as many other depths aforementioned, does.

\begin{theorem}
\label{th:GZD}
  For given $\mathcal{X}^n \subset \mathbb{R}^d ~(d \ge 1)$, suppose $\hat{\Sigma}_n$ is a positive definite matrix. Then we have: (P1). \textit{Affine-invariance}. For any $d\times d$ nonsingular matrix $\mathbf{A}$ and $d$-vector $\b$, we have that $\LD_\infty(\mathbf{A}\x + \b, P_{n, \mathbf{A} X+\b}) = \LD_\infty(\x, P_n)$. (P2). \textit{Maximality at $\bar{X}$}. $\LD_\infty (\bar{X}, P_n) = \sup\limits_{\x \in \mathbb{R}^{d}} \LD_\infty (\bar{X}, P_n)$. (P3). \textit{Monotonicity relate to $\bar{X}$}. $\LD_\infty (\x, P_n) \leq  \LD_\infty (\bar{X} + \lambda  (\x - \bar{X}), P_X)$ holds for any $\lambda \in [0, 1]$. (P4). \textit{Vanishing at infinity}. $\LD_\infty (\x, P_n) \rightarrow 0$ as $\|\x\| \rightarrow \infty$. Here $P_{n, \mathbf{A} X+\b}$ denotes the empirical probability measure related to $\{\mathbf{A} X_1 + \b, \mathbf{A} X_2 + \b, \cdots, \mathbf{A} X_n + \b\}$.
\end{theorem}

Remarkably, the assumption that $\hat{\Sigma}_n$ is a positive definite matrix holds in probability one when the covariance matrix cov$(X)$ of $X$ is positive. For the case of cov$(X)$ being singular, we recommend to reduce the dimensionality of the data in advance, and then build the statistical depth functions in the lower dimensional space.

By observing that
\begin{eqnarray*}
  \lim_{q \rightarrow +\infty} \sqrt[q]{\frac{1}{n} \sum_{i=1}^n |np_i - 1|^q} = d_{\infty} (n\p, n\p_0),
\end{eqnarray*}
we may further extend $\LD_\infty(\x, P_n)$ to the following version
\begin{eqnarray}
\label{eqn:LDq}
  \LD_q(\x, P_n) = \frac{1}{1 + S_q(\x, P_n)},
\end{eqnarray}
for some $q \in [1, +\infty]$, where
\begin{eqnarray*}
  S_q(\x, P_n) = \inf_{\p \in \mathbb{P}_x} d_q(n\p, n\p_0)
\end{eqnarray*}
with
\begin{eqnarray*}
  d_q(n\p, n\p_0) = \sqrt[q]{\frac{1}{n} \sum_{i=1}^n |np_i - 1|^q}.
\end{eqnarray*}

The class of $L_q$-norm zonoid depths defined here is quite fruitful. Similar to Theorem \ref{th:GZD}, we can show that $\LD_q(\x, P_n)$ satisfies the definition of the general statistical depth function suggested by \cite{ZS2000}.

Furthermore, based on the discussion on the Euclidean likelihood in \cite{Owe2001}, we directly derive the following result.
\begin{proposition}
\label{pro:ZDqvsMD}
  For $q = 2$, the $L_q$-norm zonoid depths defined in \eqref{eqn:LDq} deduce to the Mahalanobis depth. That is,
  \begin{eqnarray*}
    \LD_2(\x, P_n) = \frac{1}{1 + \sqrt{(\x - \bar{X})^\top \hat{\Sigma}_n^{-1} (\x - \bar{X})}} = \MD(\x, P_n).
  \end{eqnarray*}
\end{proposition}

In this sense, Proposition \ref{pro:ZDqvsMD} indicates that the class of $L_q$-norm zonoid depths defined in this paper include the Mahalanobis depth as a special case.

To facilitate the theoretical derivation, it is desirable to given the population version of the defined statistical depth function. For the $L_q$-norm zonoid depths above, the population versions are given as follows. That is,
\begin{eqnarray*}
  \LD_{q}(\x, P) = \frac{1}{1 + \inf_{g \in \mathcal{G}_x} d_q^F (g, 1)},
\end{eqnarray*}
where $\mathcal{G}_x := \left\{g(\cdot): \int Xg(X) dP = \x,~ \int g(X) dP = 1\right\}$, i.e., the family of functions satisfying both $\int Xg(X) dP = \x$ and $\int g(X) dP = 1$, and $d_q^F (g, 1) = \sqrt[q]{\int |g(X) - 1|^q dP}$ denotes the $L_q$-norm for functions of $g(\x) - 1$.

The following theorem shows the convergence of the sample $L_q$-norm zonoid depth to its population version for any $q \in [1, +\infty]$ under some regular conditions.
\begin{theorem}
\label{th:convergence}
  Suppose that $X_1, X_2, \cdots, X_n$ are i.i.d. copies of $X$ such that $E|X| < +\infty$ and cov$(X) > 0$. Then we have
  \begin{eqnarray*}
    \LD_{q}(\x, P_n) \stackrel{p}{\longrightarrow} \LD_{q}(\x, P),\quad \text{ for any } \x \in \mathbb{R}^d,
  \end{eqnarray*}
  for $q \in [1, +\infty]$, where $\stackrel{p}{\longrightarrow}$ denotes the convergence in probability.
\end{theorem}

It is worth mentioning that since the definition of $\LD_q(\x, P_n)$ involves the absolute sign, it is impossible to use directly the Lagrange multiple method as usually did for empirical likelihood in \cite{Owe2001} to prove this theorem. The absolute sign brings greatly additional difficulties to the proof; see the Appendix for details.

Theorem \ref{th:convergence} indicates that the proposed $L_q$-norm zonoid depths have nonsingular population versions. Similar to \cite{ZS2000}, we may define the $\alpha$-trimmed depth region as:
\begin{eqnarray*}
  \{\x: \LD_q(\x, P_n) \ge \alpha\},
\end{eqnarray*}
and its boundary to be the $\alpha$-th $L_q$-norm zonoid depth contour for some $\alpha \in (0, 1]$. Following the same line as Theorem \ref{th:GZD} and \cite{KM1997}, it is easy to check that these trimmed depth regions are also bounded, convex and nested for any $\alpha \in (0, 1]$, and converge to its population versions in probability in Hausdorff distance. We omit the details in this paper.

\vskip 0.1 in
\section{Illustrations}
\paragraph{}
\vskip 0.1 in\label{Sec:Illustrations}

To be useful, a depth notion is expected to be computable. Hence, in the sequel we will discuss briefly the computational issue related to the $L_q$-norm zonoid depths at the very beginning. Based on this discussion, a few illustrations are provided to help readers to gain more insight into the proposals.
\bigskip

\emph{$\diamond$ The computational issue}. When $q = 2$, $L_2$-norm zonoid depth is actually the Mahalanobis depth, its computation is easy to achieve. For $q = 4, 6, \cdots$, the computation of $L_q$-norm zonoid depths can be transformed into a convex optimization problem. 
By \cite{BV2004}, it can be resolved by using Newton's method with equality constraints; see also \cite{Owe2001}.

For the scenarios of $q \ge 1$ but $\neq 2, 4, \cdots$, the computation is more complex, because the objective function $d_q(n\p, n\p_0) = \sum_{i=1}^n |np_i - 1|^q$ involves the absolute sign, although it is still a convex function with respect to $\p$. Let $v_i^+ = \max\{np_i - 1, 0\}$ and $v_i^- = \max\{1 - np_i, 0\}$. The computation for these scenarios is equivalent to the following nonlinear convex optimization problem:
\begin{eqnarray}
\label{eqn:NonCO}
  &&\min_{\vv} f(\vv) = \sum_{i=1}^n (v_i^+ + v_i^-)^q\\
  &&\text{subject to } \begin{cases}
    \sum\limits_{i=1}^n (v_i^+ - v_i^-) = 0\\[3ex]
    \sum\limits_{i=1}^n (v_i^+ - v_i^-) (X_i - \x) = n(\x - \bar{X})\\[3ex]
    v_i^+, v_i^- \ge 0, \text{ for } i = 1, 2, \cdots, n,
  \end{cases}\nonumber
\end{eqnarray}
where $\vv = (v_1^+, \cdots, v_n^+, v_1^-, \cdots, v_n^-)$. Primal-dual interior-point methods may be used to resolve this problem \citep{BV2004}. In Matlab, the function \emph{fmincom.m} can used to resolve these problems.

Specially, for $q = 1$, the objective function in \eqref{eqn:NonCO} is linear, hence \eqref{eqn:NonCO} can be resolved by typically using the technique of linear programming. For the degenerate case $q = +\infty$, the computation can be transformed into the linear programming problem as follows:
\begin{eqnarray*}
  &&\min_{t, \vv} f(t, \vv) = t\\
  &&\text{subject to } \begin{cases}
    t \ge v_i^+ + v_i^-, ~i = 1, 2, \cdots n,\\[3ex]
    \sum\limits_{i=1}^n (v_i^+ - v_i^-) = 0\\[3ex]
    \sum\limits_{i=1}^n (v_i^+ - v_i^-) (X_i - \x) = n(\x - \bar{X})\\[3ex]
    t \ge 0, ~v_i^+, v_i^- \ge 0, \text{ for } i = 1, 2, \cdots, n.
  \end{cases}
\end{eqnarray*}
In Matlab, some well developed function, e.g., \emph{linprog.m}, can be utilized to fulfil these tasks.

\emph{$\diamond$ Illustrations}. In this part we will give some illustrations about the $L_q$-norm zonoid depths. The first example is based on a real data set taken from Example 4 in Page 57 of \cite{RouLer1987}. This data set are widely used in illustrating the depth contours. It consists of 28 animals' brain weight (in grams) and body weight (in kilograms). We use the logarithms of these data in this paper.

\begin{figure}[H]
\centering
	\subfigure[$L_1$-norm zonoid depth contours]{
	\includegraphics[angle=0,width=2.9in]{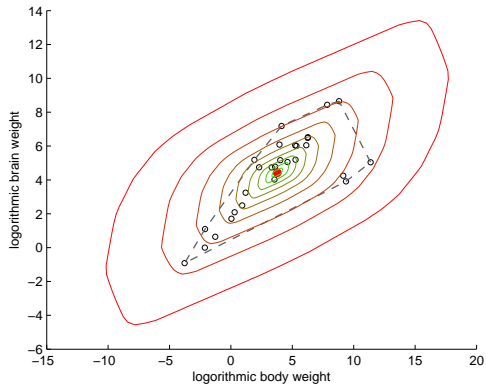}
	}\quad
	\subfigure[$L_2$-norm zonoid depth (Mahalanobis depth) contours]{
	\includegraphics[angle=0,width=2.9in]{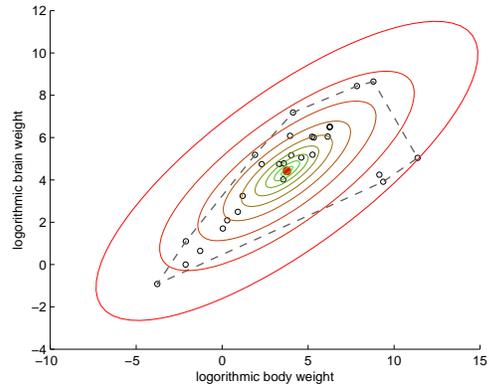}
	}
    \subfigure[$L_4$-norm zonoid depth contours]{
	\includegraphics[angle=0,width=2.9in]{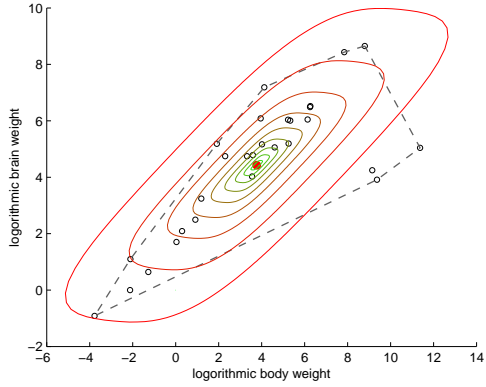}
	}\quad
    \subfigure[$L_8$-norm zonoid depth contours]{
	\includegraphics[angle=0,width=2.9in]{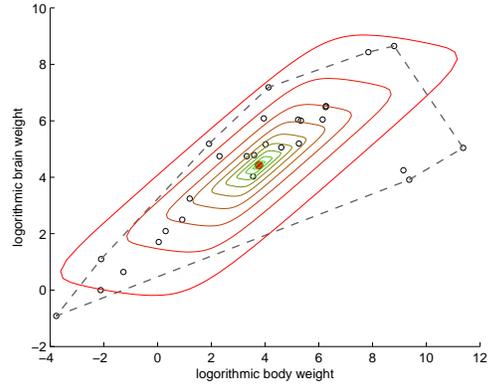}
	}
\caption{Shown are contours induced from the $L_q$-norm zonoid depth. Here the small hollow points stand for the observations, and the big solid point is the sample mean. The dashes line stands for the boundary of the convex hull of the data cloud.}
\label{fig:Real}
\end{figure}

Figure~\ref{fig:Real} reports ten depth contours with depth values $0.2500, 0.3333, \cdots, 1.0000$ from the
periphery inwards for the $L_q$-norm zonoid depths with $q = 1, 2, 4, 8$, respectively. As shown from this figure, $L_2$-norm zonoid depth induces some elliptical contours, which confirms Proposition \ref{pro:ZDqvsMD}. All of these depth functions center at the sample mean, and have positive depth values outside the convex hull of the data set, i.e., the area formed by dashes lines.

Besides this, we also provide some illustration for these depths based on the simulated data. Three scenarios are considered here, namely, (S1) $X \sim U([-1, 1]\times[-1, 1])$, (S2) $X \sim N(\0, \mathbf{I}_{2\times 2})$, and (S3) $X = (Y^2 + Z, Z^2 + Y)^\top$ with $Y \sim N(0, 1)$ and $Z \sim N(0, 1)$. The sample size are 1000 for all of these scenarios. The support of (S1) is a finite set. The distribution in (S2) is symmetrical, while that in (S3) is skewed. The depth values of these contours are $0.2000,    0.2889, \cdots, 1$ generated by \emph{linspace(0.2, 1, 10)}. See Figures~\ref{fig:Unif}-\ref{fig:Skew} for details, where the small points stand for the data points, and the big points are their sample means.

\begin{figure}[H]
\centering
	\subfigure[$L_1$-norm zonoid depth contours]{
	\includegraphics[angle=0,width=2.9in]{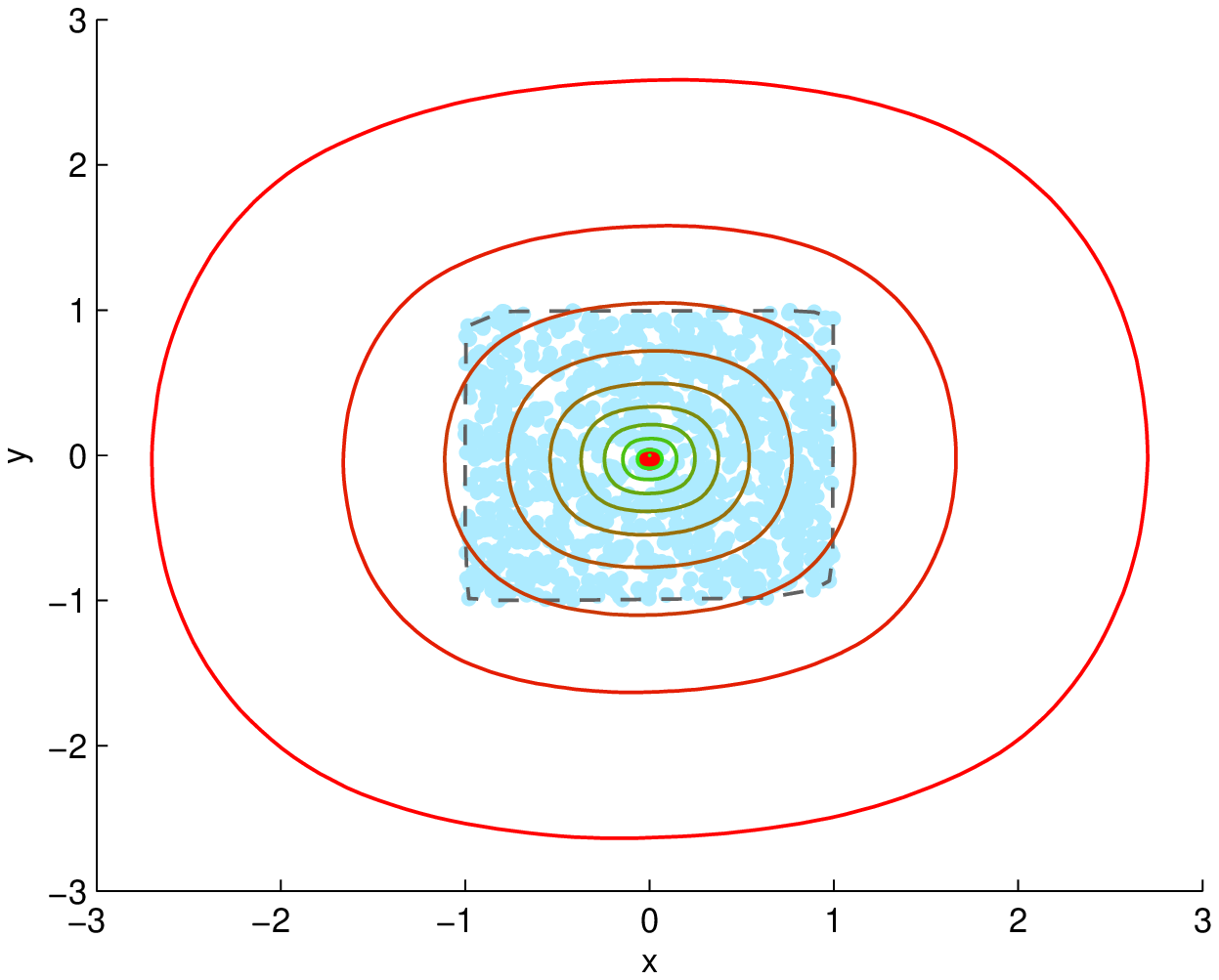}
	}\quad
	\subfigure[$L_2$-norm zonoid depth contours]{
	\includegraphics[angle=0,width=2.9in]{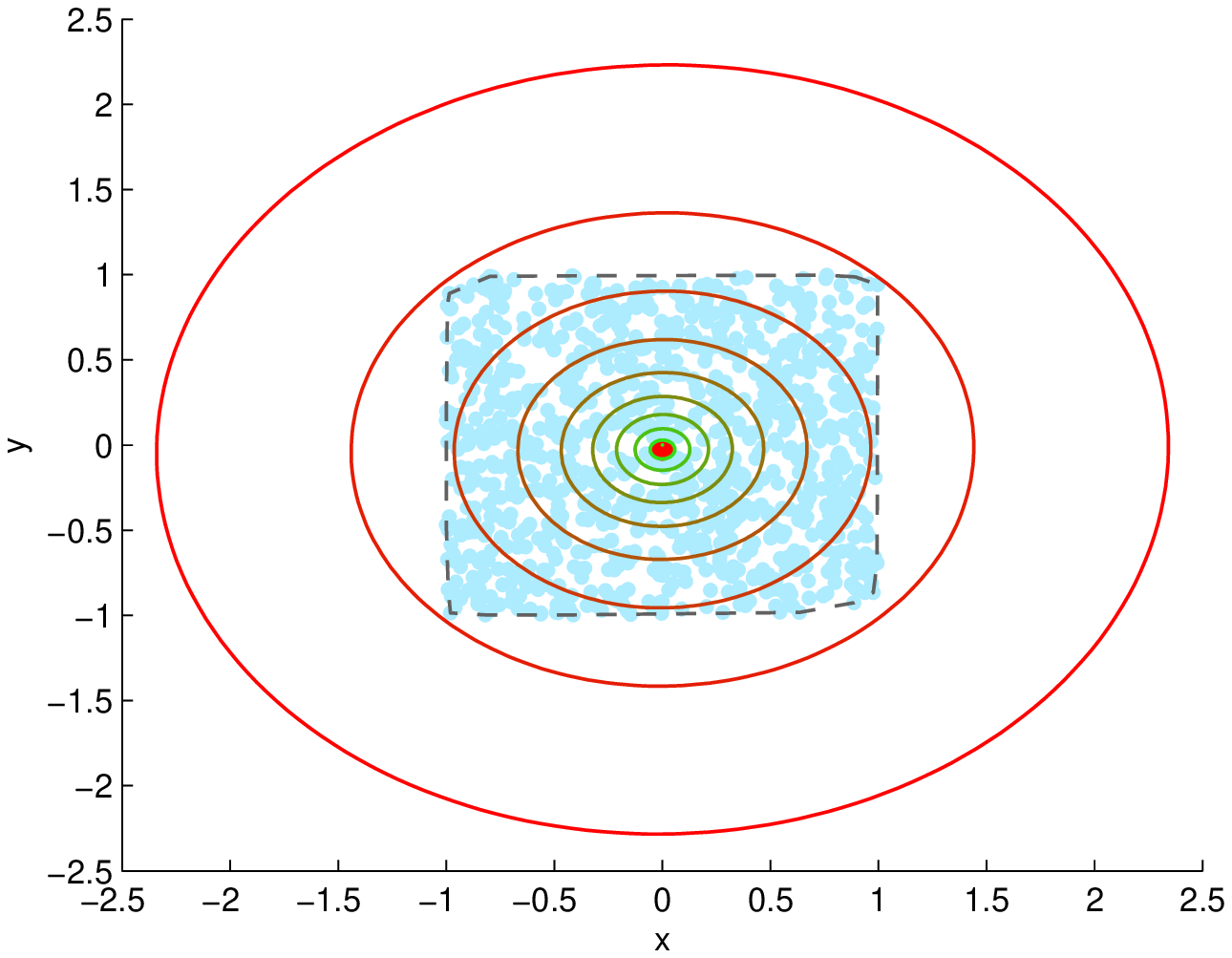}
	}
    \subfigure[$L_4$-norm zonoid depth contours]{
	\includegraphics[angle=0,width=2.9in]{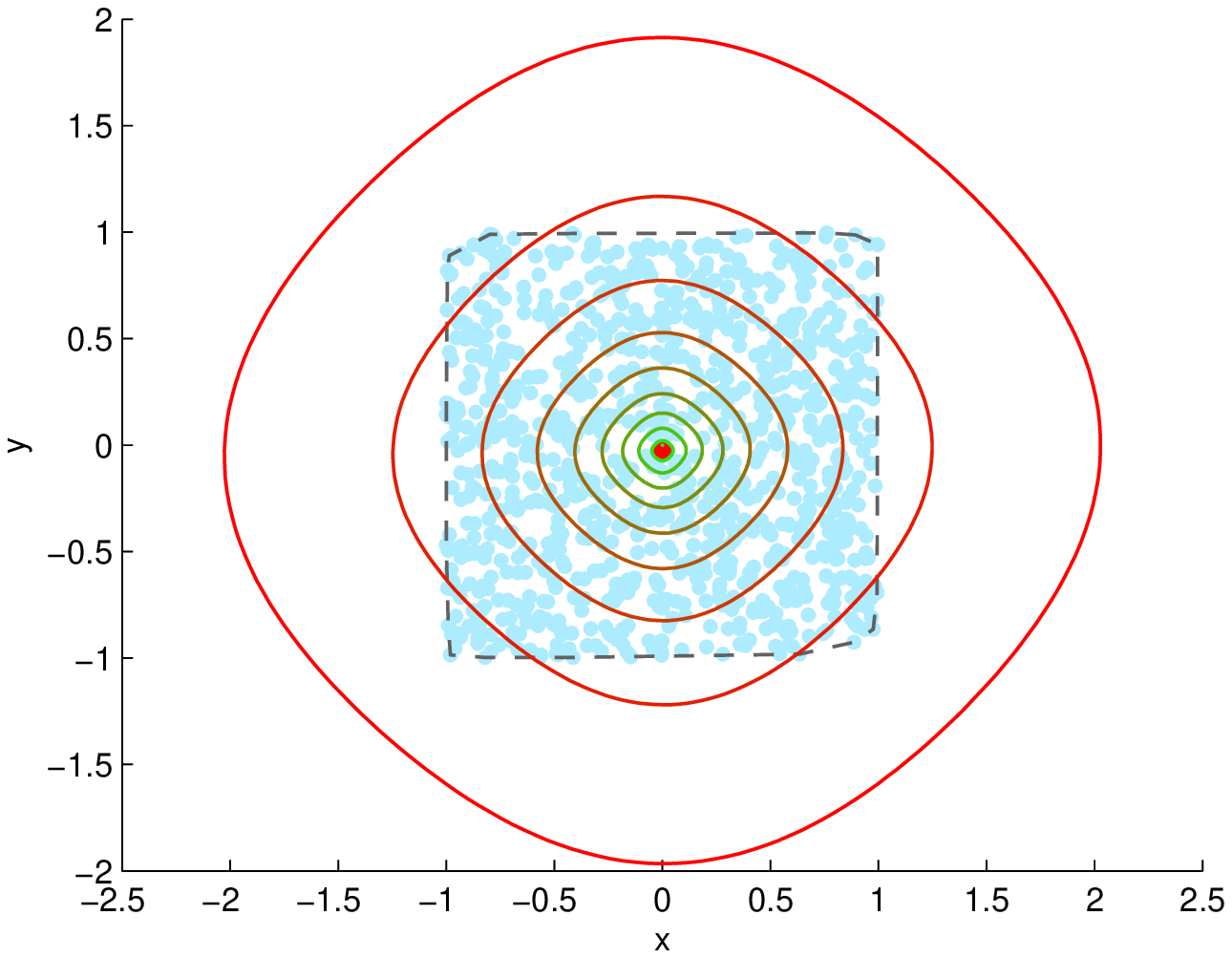}
	}\quad
    \subfigure[$L_8$-norm zonoid depth contours]{
	\includegraphics[angle=0,width=2.9in]{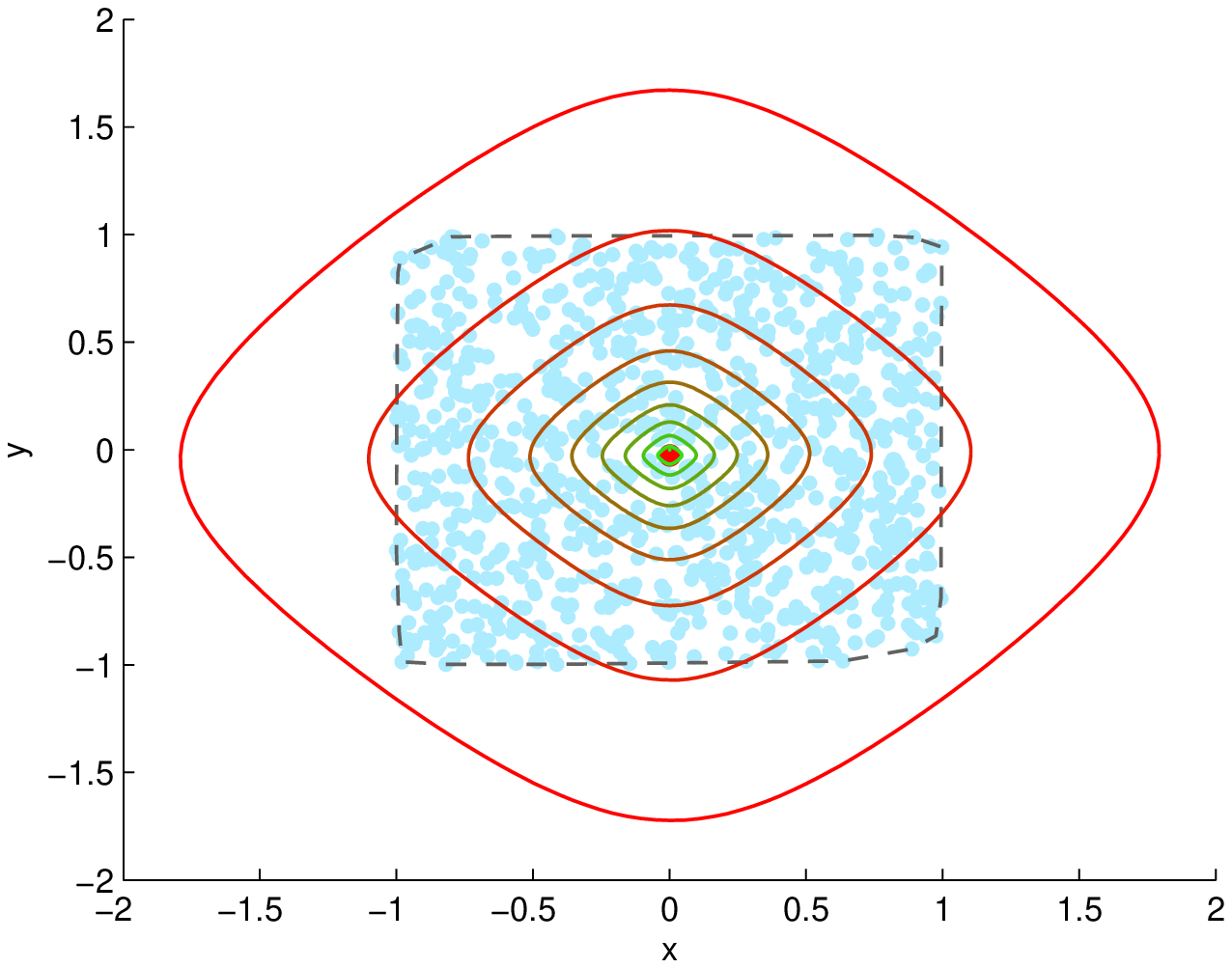}
	}
\caption{Shown are contours for $L_q$-norm zonoid depths related to Scenario (S1).}
\label{fig:Unif}
\end{figure}

Among these $L_q$-norm zonoid depths, the $L_1$-norm zonoid depth appears to be very desirable. As shown in Figures~\ref{fig:Real}-\ref{fig:Skew}, the shapes of the contours of the $L_1$-norm zonoid depth are not fixed and appear to be dependent on the data. Specially, its contours are roughly squared for the data set generated uniformly from $[-1, 1]\times [-1, 1]$, whose shape is squared. While for Scenario (S2), the shape of contours related this depth is roughly spherical. Hence, we recommend practitioners to use this depth in the applications, e.g., classification, as an alterative to the conventional zonoid depth.

\begin{figure}[H]
\centering
	\subfigure[$L_1$-norm zonoid depth contours]{
	\includegraphics[angle=0,width=2.8in]{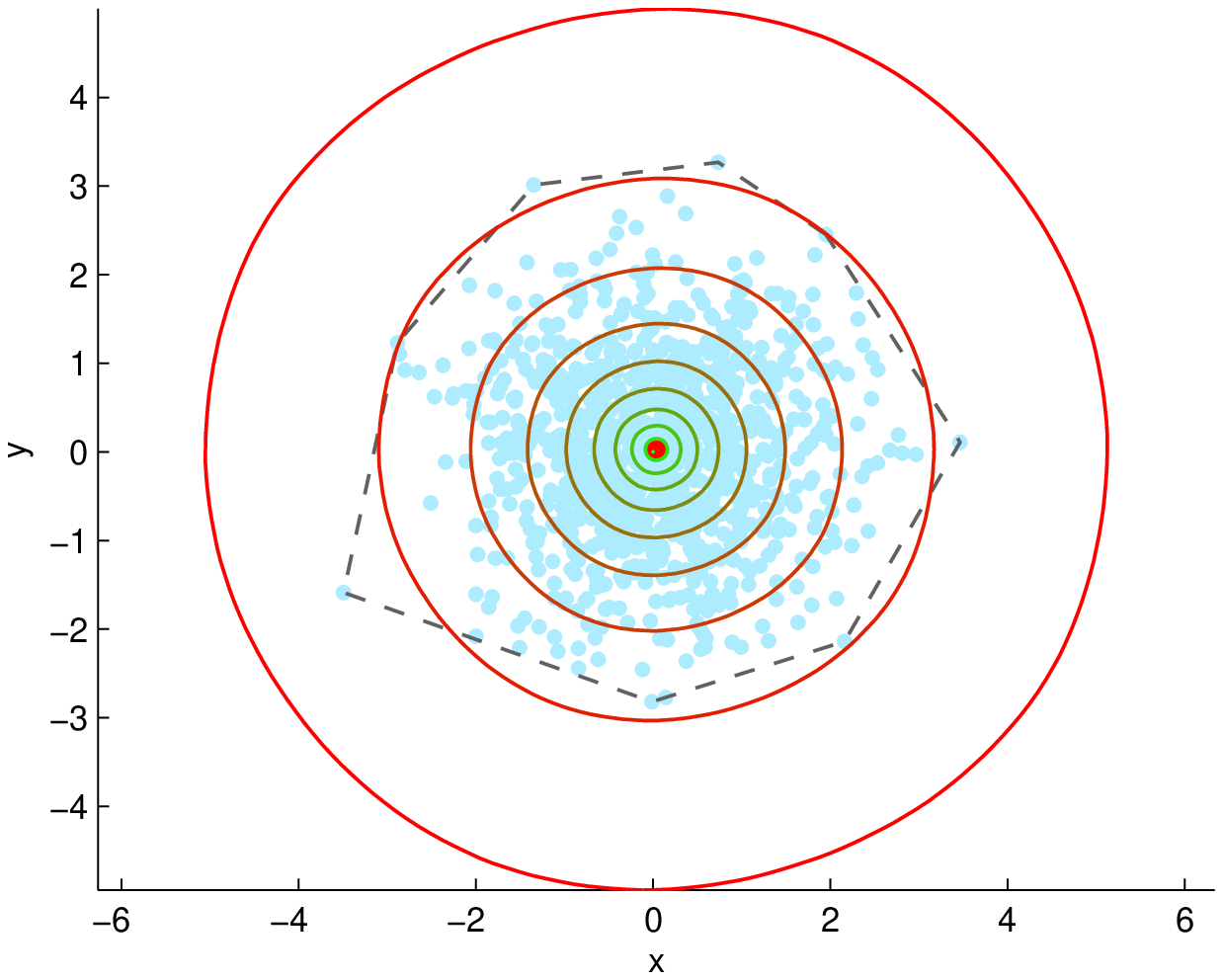}
	}\quad
	\subfigure[$L_2$-norm zonoid depth contours]{
	\includegraphics[angle=0,width=2.8in]{NormQ1.eps}
	}
    \subfigure[$L_4$-norm zonoid depth contours]{
	\includegraphics[angle=0,width=2.8in]{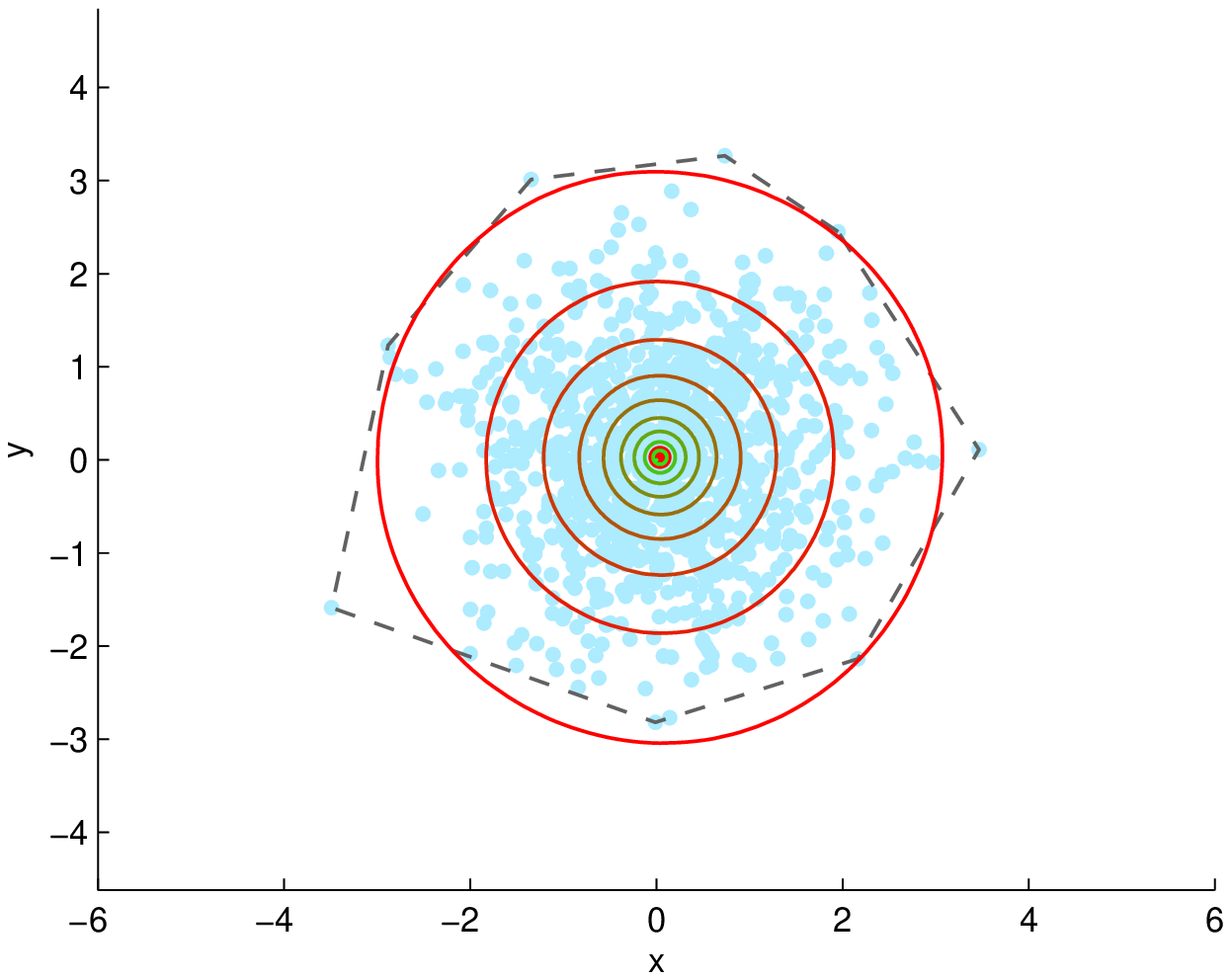}
	}\quad
    \subfigure[$L_8$-norm zonoid depth contours]{
	\includegraphics[angle=0,width=2.8in]{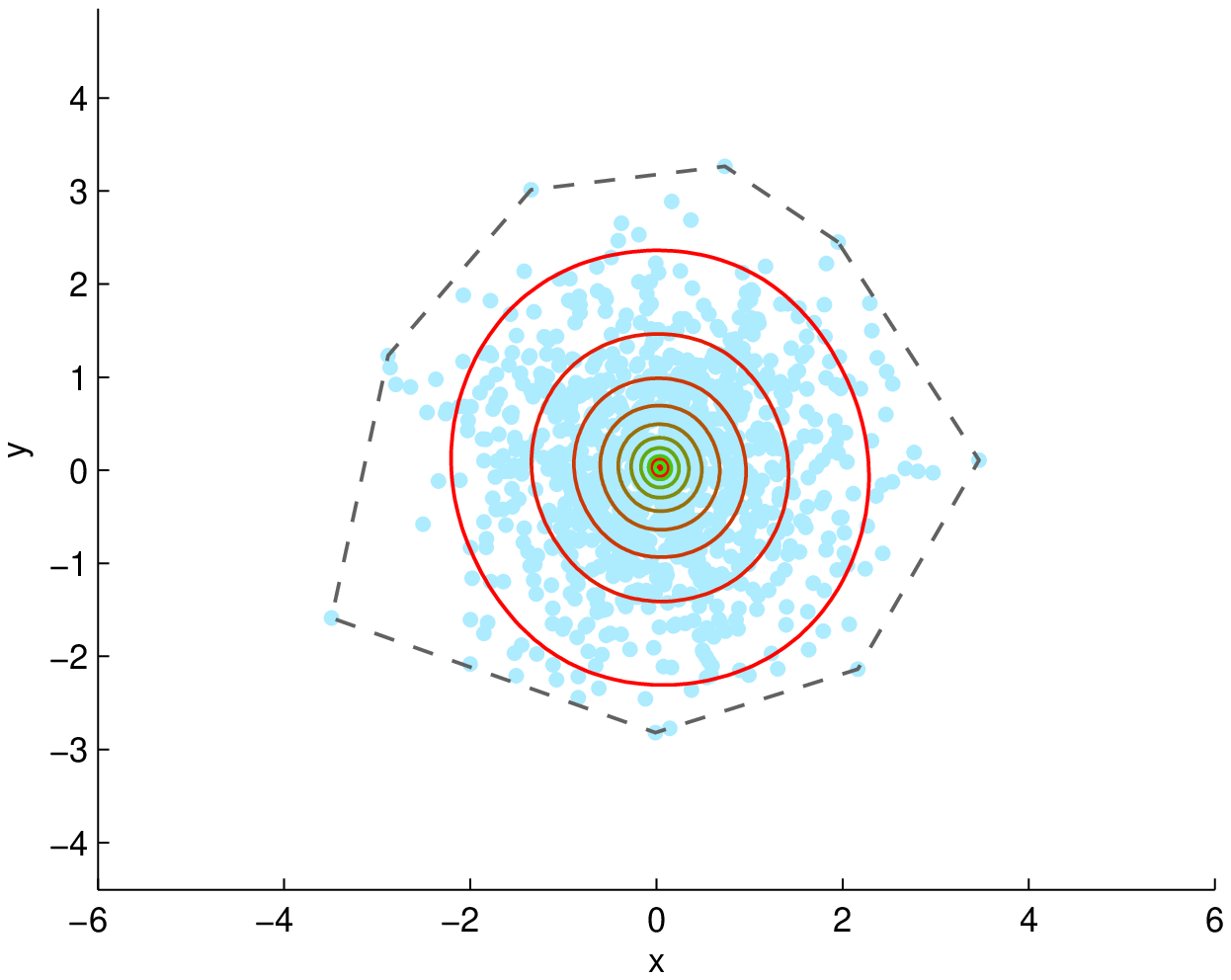}
	}
\caption{Shown are contours for $L_q$ zonoid depth related Scenario (S2).}
\label{fig:Norm}
\end{figure}

\vskip 0.1 in
\section{Concluding remarks}
\paragraph{}
\vskip 0.1 in

In this paper, we proposed a new class of statistical depth functions, which still relate to the so-called lift zonoid. Different from the conventional zonoid depth defined in \cite{KM1997}, we based the new depth functions on the $L_q$-norm between $\p$ (related to $\x$) and $\p_0$ (related to the sample mean). Since $L_q$ norm allows the components of $\p$ to be negative, the new depth functions take positive value even outside the data cloud. The data examples shows that the $L_1$-norm zonoid depth function appears to very favorite. We anticipate it to be helpful in the practical applications if the outside problem is a problem having to be taken into account.

\begin{figure}[H]
\centering
	\subfigure[$L_1$-norm zonoid depth contours]{
	\includegraphics[angle=0,width=2.9in]{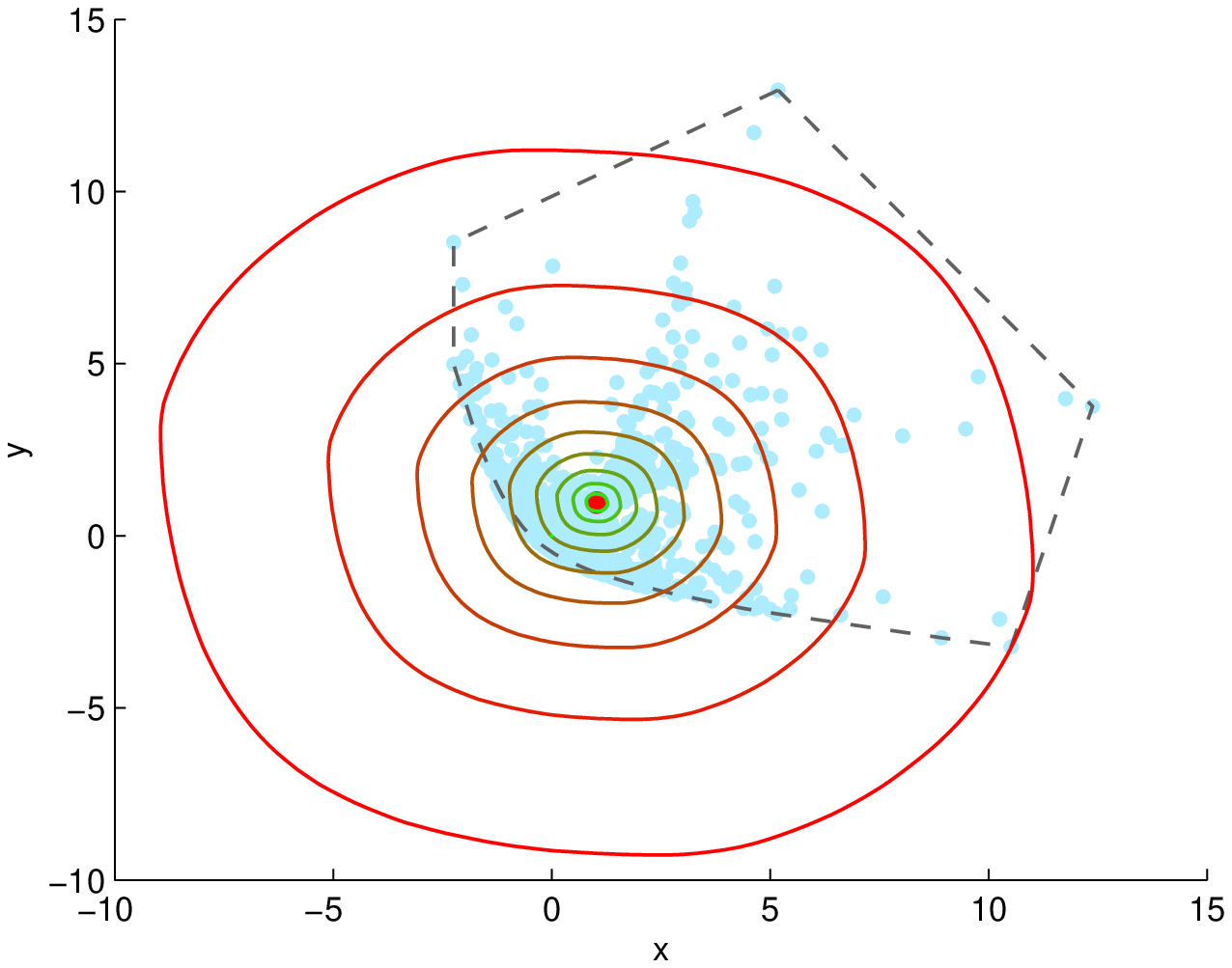}
	}\quad
	\subfigure[$L_2$-norm zonoid depth contours]{
	\includegraphics[angle=0,width=2.9in]{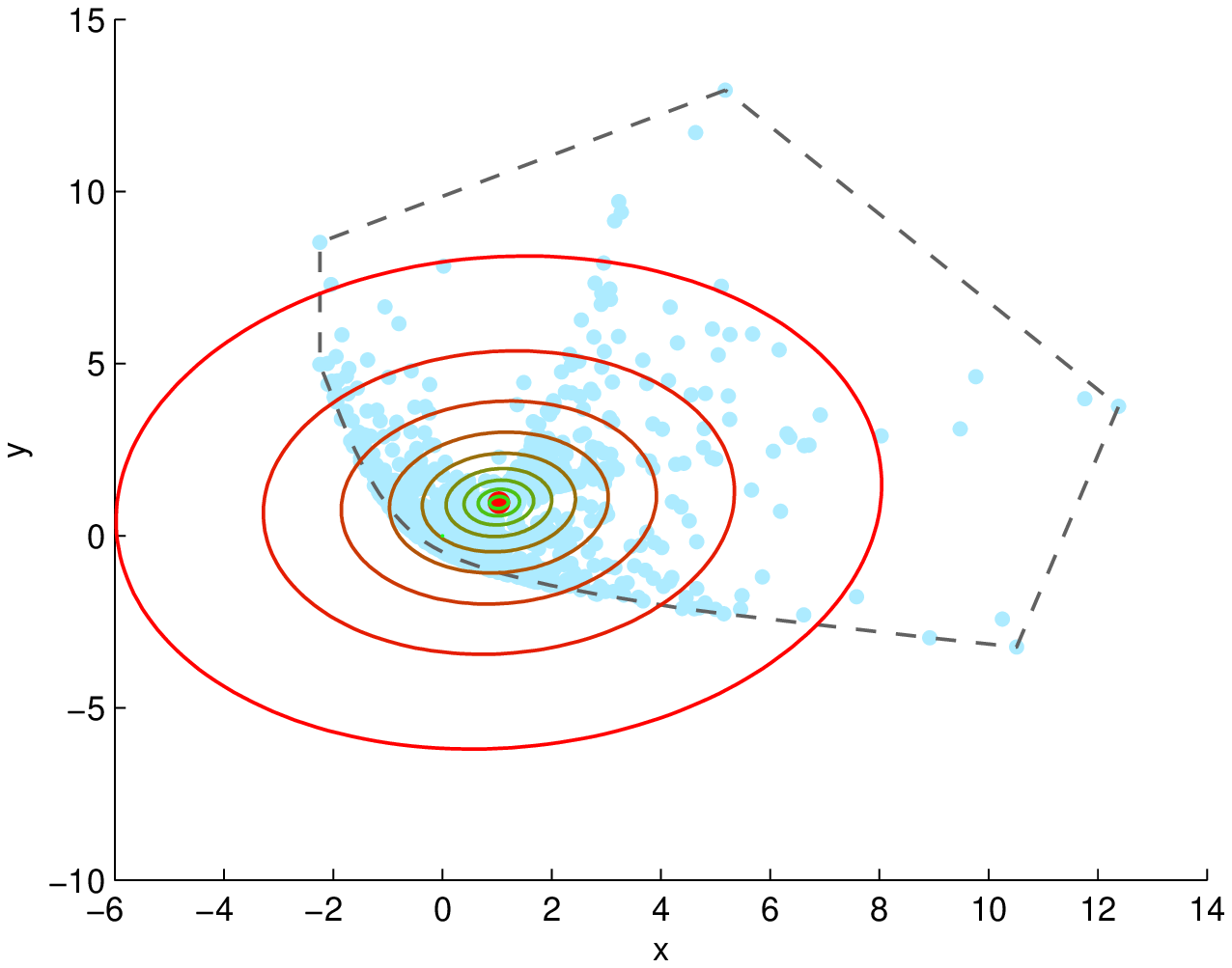}
	}
    \subfigure[$L_4$-norm zonoid depth contours]{
	\includegraphics[angle=0,width=2.9in]{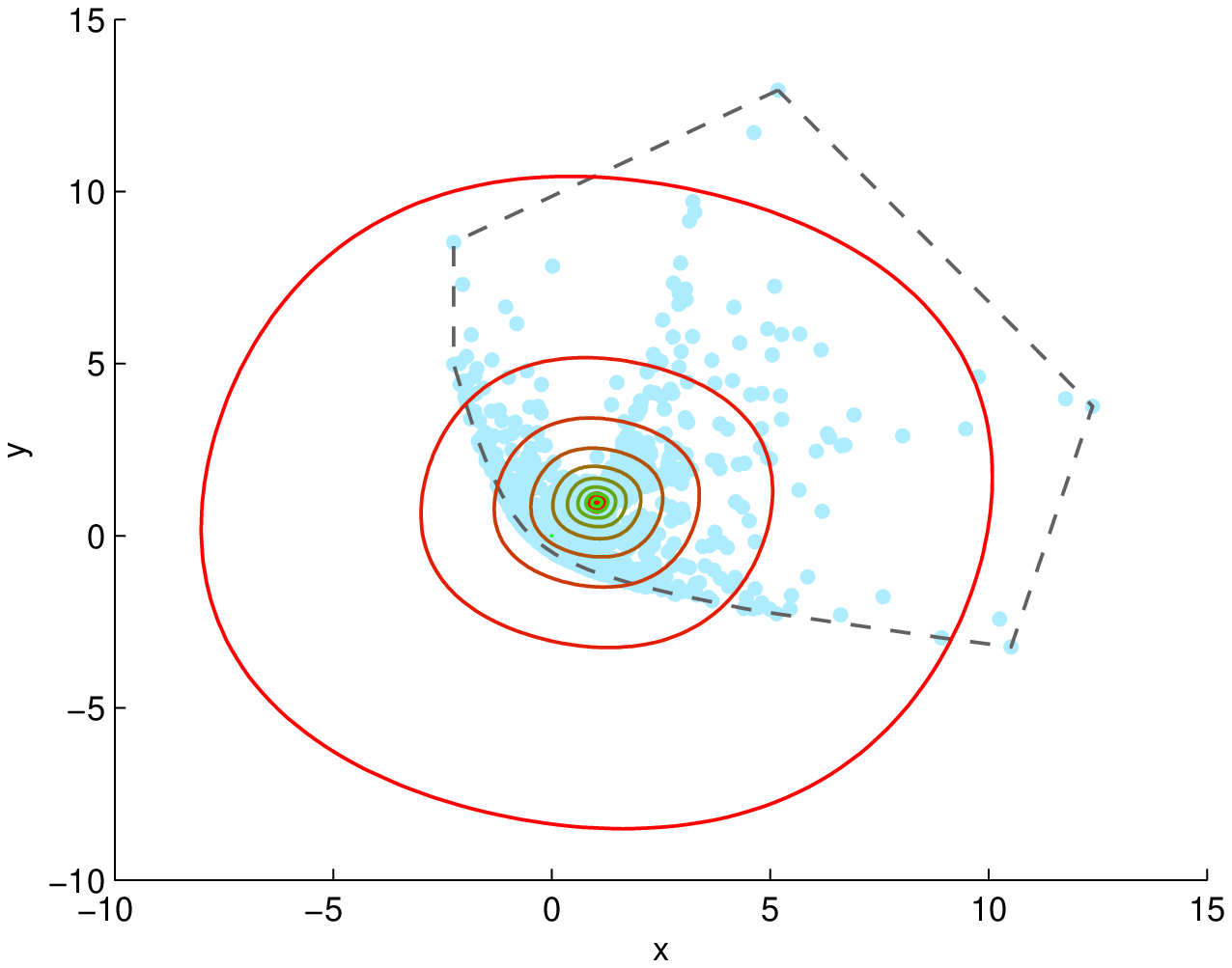}
	}\quad
    \subfigure[$L_8$-norm zonoid depth contours]{
	\includegraphics[angle=0,width=2.9in]{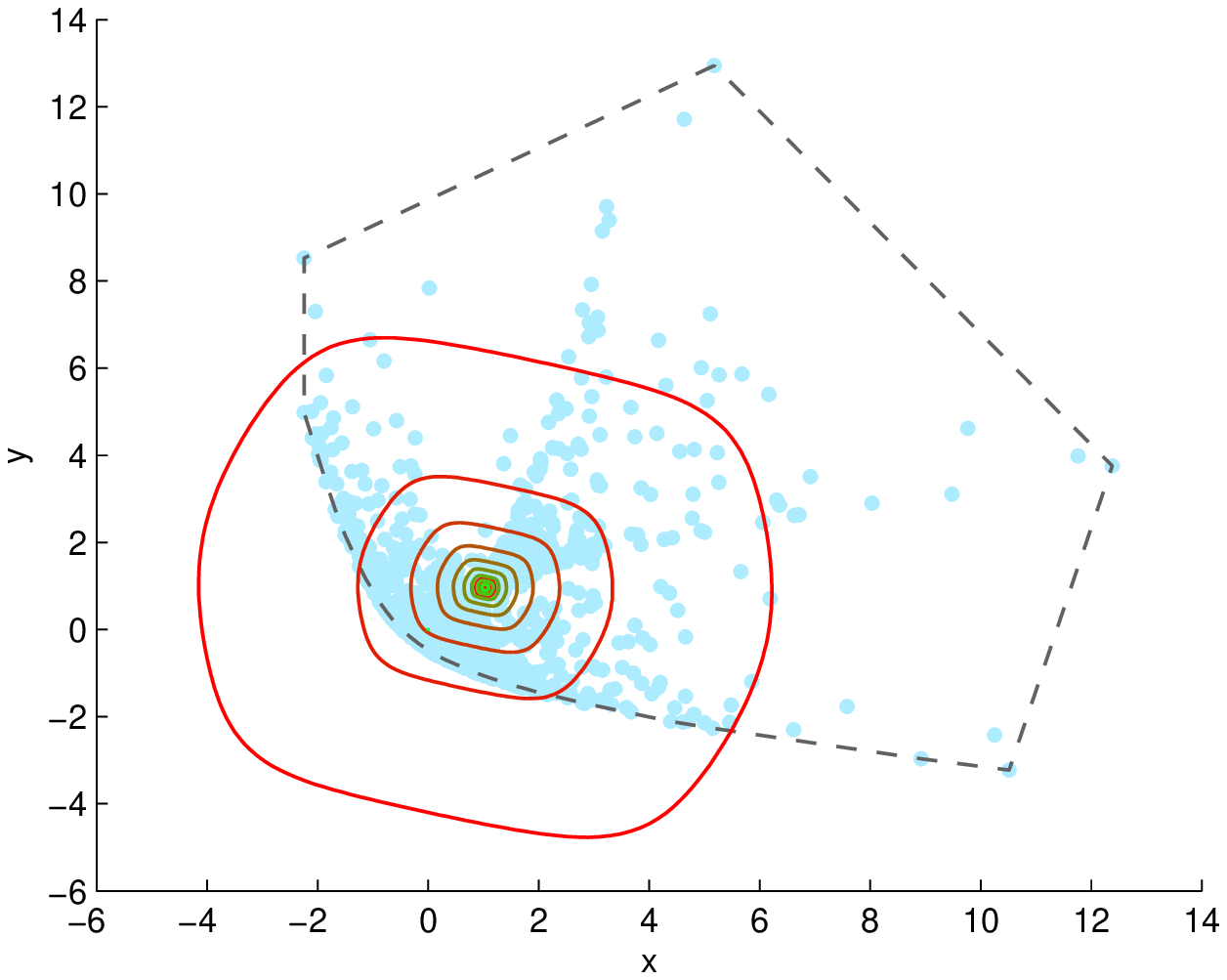}
	}
\caption{Shown are contours for $L_q$-norm depths related to Scenario (S3).}
\label{fig:Skew}
\end{figure}

\vskip 0.1 in
\section*{Acknowledgements}
\paragraph{}
\vskip 0.1 in

The research is supported by NNSF of China (Grant No.11601197, 11461029), China Postdoctoral Science Foundation funded project (2016M600511, 2017T100475),  NSF of Jiangxi Province (No.20171ACB21030, 20161BAB201024), and the Key Science Fund Project of Jiangxi provincial education department (No.GJJ150439).
\medskip

\vskip 0.1 in
\section*{Appendix: Detailed proofs of the main results}\label{Sec:Appendix}
\paragraph{}
\vskip 0.1 in

In this appendix, we provided the detailed proofs of the main proposition and theorems.

\begin{proof}[Proof of Theorem \ref{th:GZD}]
  The check of Properties P1-P2 is trivial. We omit the details.

  \emph{Property P3}. Let $\y = \bar{X} + \theta (\x - \bar{X}) = (1 - \theta) \bar{X} + \theta \x$. Let $\p_x^* \in \mathbb{P}_x$ be the weight vector satisfying
  \begin{eqnarray*}
    d_\infty(n\p_x^*, n\p_0) = \inf_{\p_x \in \mathbb{P}_x} d_\infty(n\p_x, n\p_0).
  \end{eqnarray*}
  Denote $\bar{\mathbb{P}}_x = \mathbb{P}_x \cap \{\p \in \mathbb{R}^d: \|\p - \p_x^*\| < \delta\}$ for some $\delta > 0$, where $\|\x\| = \sqrt{\x^\top \x}$. Clearly, $\bar{\mathbb{P}}_x$ is a close and bounded set. Denote $\bar{\mathbb{P}}_{y, x} = \{\p \in \mathbb{R}^d: \p = (1 - \theta) \p_0 + \theta \p_x, ~\forall \p_x \in \bar{\mathbb{P}}_x\}$. By noting that for any $\forall \gamma \in (0, 1)$, and any $\p_1, \p_2$ on the boundary of $\bar{\mathbb{P}}_{y, x}$, it holds
  \begin{eqnarray*}
    d_\infty(n((1 - \gamma) \p_1 + \gamma \p_2), n\p_0) \le (1 - \gamma) \cdot d_{\infty}(n\p_1, n\p_0) + \gamma \cdot d_{\infty} (n \p_2, n\p_0).
  \end{eqnarray*}
  This implies that $\inf_{\p\in \mathbb{P}_y} d_{\infty}(n\p, n\p_0) = \inf_{\p\in \mathbb{P}_{y,x}} d_{\infty}(n\p, n\p_0)$ by the convexity of $\bar{\mathbb{P}}_{y, x}$. Using this and the convexity of $d_{\infty}(\cdot, n\p_0)$ for fixed $\p_0$, it is easy to show Property P3 based on the construction of $\bar{\mathbb{P}}_{y, x}$.

  \emph{Property P4}. For any $\p \in \mathbb{P}_x$, observe that $\x = \mathbf{A}_X \p$. Hence, there must exist a orthogonal matrix $\mathbf{U}$ satisfying
  \begin{eqnarray*}
    \|\x\|^2 = \x^\top \x &=& \p^\top \mathbf{A}_X^\top \mathbf{A}_X \p\\
    &=& (\mathbf{U}\p)^\top (\mathbf{U} \mathbf{A}_X^\top \mathbf{A}_X \mathbf{U}^\top) (\mathbf{U} \p)\\
    &=& (\mathbf{U}\p)^\top
    \left(
    \begin{array}{ccccccc}
      \lambda_1 & & & & & &\\
       & \lambda_2& & & & &\\
       && \ddots & & & &\\
       && &\lambda_d & & &\\
       && & &0 & &\\
       && & & &\ddots &\\
       && & & &&0
    \end{array}
    \right)
    (\mathbf{U} \p)\\
    &\le& \lambda_{\max} \cdot \|\p\|^2 \rightarrow +\infty, \text{ as } \|\x\| \rightarrow +\infty,
  \end{eqnarray*}
  where $\lambda_{\max} = \max\{\lambda_1, \cdots, \lambda_d\}$ with $\lambda_1, \cdots, \lambda_d$ being the positive eigenvalues of $\mathbf{A}_X^\top \mathbf{A}_X$, when $\hat{\Sigma}_n$ is positive.

  Next, by noting that, for some $q \ge 1$,
  \begin{eqnarray*}
    \sqrt[q]{\frac{1}{n} \sum_{i=1}^n |np_i - 1|^q} &=& n^{\frac{q-1}{q}} \cdot \sqrt[q]{\sum_{i=1}^n \left|p_i - \frac{1}{n}\right|^q}\\
    &\ge&
    n^{\frac{q-1}{q}} \cdot \left(\sqrt[q]{\sum\limits_{i=1}^n \left|p_i\right|^q} - n^{\frac{1-q}{q}}\right)
  \end{eqnarray*}
  we claim that $d_{\infty}(n\p, n\p_0) \rightarrow +\infty$ as $\|\p\| \rightarrow +\infty$. Since $\p$ is any given, we in turn obtain
  \begin{eqnarray*}
    \EZD(\x, P_n) = \frac{1}{1 + \inf\limits_{\x \in \mathbb{P}_x} d_{\infty}(n\p, n\p_0)} \rightarrow 0, \text{ as } \|\x\| \rightarrow +\infty.
  \end{eqnarray*}
  This completes the proof of this theorem.
\end{proof}

\begin{proof}[Proof of Theorem \ref{th:convergence}]
  For any $\x \in \mathbb{R}^d$, let $\mathcal{G}_x := \left\{g(\cdot): \int Xg(X) dP = \x,~ \int g(X) dP = 1\right\}$, and $\mathbb{P}_{x,n} = \{(p_{n1}, p_{n2}, \cdots, p_{nn})^\top: \sum_{i=1}^n p_{ni} X_i = \x, ~\sum_{i=1}^n p_{ni} = 1\}$.

  By the construction of $\mathcal{G}_x$, it is easy to check that $\mathcal{G}_x$ is a \emph{close and convex} set. This, together with the fact that $d_q^F(\cdot, 1)$ is a convex function for any $q \in [1, +\infty]$, directly result in that there must exist a $g_0(\cdot) \in \mathcal{G}_x$ satisfying
  \begin{eqnarray*}
    d_q^F(g_0, 1) = \inf_{g\in \mathcal{G}_x} d_q^F(g, 1).
  \end{eqnarray*}

  For given i.i.d. samples $X_1, X_2, \cdots, X_n$, let $\p_n = (p_{n1}, p_{n2}, \cdots, p_{nn})^\top$, where
  \begin{eqnarray*}
    p_{ni} = \frac{g_0(X_i)}{\sum\limits_{j=1}^n g_0(X_j)},~ i = 1, 2, \cdots, n.
  \end{eqnarray*}
  Observe that
  \begin{eqnarray}
  \label{eqn:slutsky}
    \sum_{i=1}^n p_{ni} = 1, \quad \text{and} \quad \sum_{i=1}^n p_{ni} X_i \stackrel{p}{\longrightarrow} E(g_0(X)X) = \x, \text{ as } n \rightarrow +\infty
  \end{eqnarray}
  by Slutsky's lemma. Denote $\tilde{\p}_n = (\tilde p_{n1}, \tilde p_{n2}, \cdots, \tilde p_{nn})^\top \in \mathbb{P}_{x,n}$ as the projection of $\p_n$ onto $\mathbb{P}_{x, n}$. Clearly, $\bar{\p}_n = \p_n - \tilde{\p}_n$ belongs to the space spanned by the rows of the matrix ${\mathbf{A}_X \choose \textbf{1}_n^\top}$. Hence, there must exist a coefficient vector $\bm{\beta}_n \in \mathbb{R}^{d+1}$ such that $\bar{\p}_n = (\mathbf{A}_X^\top, \textbf{1}_n) \bm{\beta}_n$. It, combined with \eqref{eqn:slutsky}, leads to
  \begin{eqnarray*}
    \left\|{\mathbf{A}_X \choose \textbf{1}_n^\top} \p_n - {\x \choose 1}\right\| &=& \left\|{\mathbf{A}_X \choose \textbf{1}_n^\top} \p_n - {\mathbf{A}_X \choose \textbf{1}_n^\top} \tilde{\p}_n \right\|\\
    &=& \left\|{\mathbf{A}_X \choose \textbf{1}_n^\top}{\mathbf{A}_X \choose \textbf{1}_n^\top}^\top \bm{\beta}_n\right\| \stackrel{p}{\longrightarrow} 0,
  \end{eqnarray*}
  as $n \rightarrow +\infty$. This implies that $\|\bm{\beta}_n\| = o_p(1)$ by noting that $\textbf{1}_n^\top \mathbf{A}_X^\top(\mathbf{A}_X\mathbf{A}_X^\top)^{-1} \mathbf{A}_X \textbf{1}_n < \textbf{1}_n^\top \textbf{1}_n$ in probability 1. Hence, $\|\p_n - \tilde{\p}_n\| = o_p(1)$, and in turn we have
  \begin{eqnarray*}
    |d_q(n\tilde{\p}_n, n\p_0) - d_q(n\p_n, n\p_0)| \stackrel{p}{\longrightarrow} 0, \text{ as } n \rightarrow +\infty.
  \end{eqnarray*}
  A direct application of this leads to
  \begin{eqnarray*}
    \inf_{\p \in \mathbb{P}_{x, n}} d_q(n\p, n\p_0) &\le& d_q (n \tilde\p_n, n\p_0)\\
     &=& d_q(n\p_n, n\p_0) + o_p(1)\\
     &=& d_q^F(g_0, 1) + o_p(1) \\
     &=& \inf_{g \in \mathcal{G}_x} d_q^F(g, 1) + o_p(1), \text{ as } n \rightarrow +\infty.
  \end{eqnarray*}
  Hence, we have
  \begin{eqnarray}
  \label{eqn:sup}
    \limsup_{m\ge n} \{\inf_{\p \in \mathbb{P}_{x, m}} d_q(m\p, n\p_0)\} \leq \inf_{g \in \mathcal{G}_x} d_q^F(g, 1) + o_p(1), \text{ as } n \rightarrow +\infty.
  \end{eqnarray}

  On the other hand, for any $n$, let
  \begin{eqnarray*}
    \p_n^0 &=& (p_{n1}^0, p_{n2}^0, \cdots, p_{ni}^0)^\top \\
    &:=& (p_{n1}^0(X_1, \cdots, X_n), p_{n2}^0(X_1, \cdots, X_n), \cdots, p_{ni}^0(X_1, \cdots, X_n))^\top \in \mathbb{P}_{x,n}
  \end{eqnarray*}
  such that $d_q(n\p_n^0, n\p_0) = \inf\limits_{\p_n \in \mathbb{P}_{x,n}} d_q(n\p_n, n\p_0)$ (by noting the convexity and closeness of $\mathbb{P}_{x, n}$). Observe that, for given i.i.d. samples $X_1, X_2, \cdots, X_n$,
  \begin{eqnarray*}
    \x &=& E\left(\sum_{i=1}^n E(p_{ni}^0(X_1, \cdots, X_n) |X_i) X_i\right)\\
     &=& E\left(\sum_{i=1}^n E(p_{ni}^0(X_1, \cdots,X_{i-1}, X, X_{i+1}, \cdots, X_n) |X) X\right)\\
    1 &=& E\left(\sum_{i=1}^n E(p_{ni}^0(X_1, \cdots, X_n) |X_i)\right)\\
     &=& E\left(\sum_{i=1}^n E(p_{ni}^0(X_1, \cdots,X_{i-1}, X, X_{i+1}, \cdots, X_n) |X)\right).
  \end{eqnarray*}
  Hence, we claim that
  \begin{eqnarray*}
    g_n(X) = \sum_{i=1}^n E(p_{ni}^0(X_1, \cdots,X_{i-1}, X, X_{i+1}, \cdots, X_n) |X) = E(np_{n1}^0|X) \in \mathcal{G}_x,
  \end{eqnarray*}
  because $p_{ni}^0(X_1, X_2, \cdots, X_n)$'s are identically distributed. 
  Without confusion, we may drop the arguments of $p_{ni}^0$'s hereafter for convenience.

  Next, similar to \eqref{eqn:sup}, we have, for $q \ge 1$,
  \begin{eqnarray*}
    P(d_q^F(g_n, 1) \ge d_q(n\p_n^0, n\p_0)) \rightarrow 1, \text{ as } n \rightarrow +\infty.
  \end{eqnarray*}
  Using this, we derive
  \begin{eqnarray*}
    &&E\left|(d_q^F(g_n, 1))^q -  (d_q(n\p_n^0, n\p_0))^q\right|\\
    &&=E\left((d_q^F(g_n, 1))^q -  (d_q(n\p_n^0, n\p_0))^q\right) + o(1)\\
    &&=E\left((d_q^F(g_n, 1))^q - \frac{1}{n} \sum_{i=1}^n |np_{ni}^0 - 1|^q \right) + o(1)\\
    &&= E\left((d_q^F(g_n, 1))^q - \frac{1}{n} \sum_{i=1}^n E(|np_{ni}^0 - 1|^q|X_i) \right) + o(1)\\
    &&\le E\left((d_q^F(g_n, 1))^q - \frac{1}{n} \sum_{i=1}^n |E(np_{ni}^0|X_i) - 1|^q \right) + o(1)\\
    &&\rightarrow 0, \text{ as } n \rightarrow +\infty.
  \end{eqnarray*}
  This in fact show that $d_q^F(g_n, 1) =  d_q(n\p_n^0, n\p_0) + o_p(1)$. Using this, it is then easy to check that
  \begin{eqnarray*}
    d_q(n\p_n^0, n\p_0) &\ge& \liminf_{m\ge n} d_q(m\p_m^0, m\p_0)\\
    &=& \liminf_{m\ge n} d_q^F(g_{m}, 1) + o_p(1) \\
    &\ge&  \inf_{g\in\mathcal{G}_x} d_q^F(g, 1) + o_p(1), \text{ as } n \rightarrow +\infty.
  \end{eqnarray*}
  where $g_{n_k}(\cdot)$ denotes a convergent subsequence of $g_n(\cdot)$. $\lim_{n_k} g_{n_k} \in \mathcal{G}_x$ by noting the closeness of $\mathcal{G}_x$. This actually shows that
  \begin{eqnarray*}
    \liminf_{m\ge n} \{\inf_{\p \in \mathbb{P}_{x, m}} d_q(m\p, m\p_0)\} \ge \inf_{g\in\mathcal{G}_x} d_q^F(g, 1) + o_p(1), \text{ as } n \rightarrow +\infty.
  \end{eqnarray*}
  This, together with \eqref{eqn:sup}, complete the proof of Theorem \ref{th:convergence}.
\end{proof}

\end{document}